\documentclass{amsart}
\usepackage{txfonts}
\usepackage{amsfonts}
\usepackage{amssymb}
\usepackage{mathrsfs}
\numberwithin{equation}{section}
 \newtheorem{thm}{Theorem}
 
 \newtheorem{prop}[thm]{Proposition}
 \newtheorem{lem}[thm]{Lemma}
 \newtheorem{cor}[thm]{Corollary}
 \newtheorem{defn}{Definition}
 \newtheorem{rem}{Remark}
  \newcommand{\f}{\mathbb{F}_{q}}
  \newcommand{\fl}{\mathbb{F}_{q^{l}}}
   
\begin{document}
\title{An Authentication Scheme for Subspace Codes over Network Based on Linear Codes}
\author{Jun Zhang, Xinran Li and Fang-Wei Fu}
\address{Chern Institute of Mathematics, Nankai University, Tianjin, P.R. China}
\email{zhangjun04@mail.nankai.edu.cn; xinranli@mail.nankai.edu.cn; fwfu@nankai.edu.cn}
\thanks{This research is supported by the National Key Basic Research Program of China (Grant No. 2013CB834204), and the National Natural Science Foundation of
China (Nos. 61171082, 10990011 and 60872025). The author Jun Zhang is also supproted by the Chinese Scholarship Council under the State Scholarship Fund during visiting University of California, Irvine.}
\date{}
 \maketitle
\begin{abstract}
Network coding provides the advantage of maximizing the usage of network resources, and has great application prospects in future network communications. However, the properties of network coding also make the pollution attack more serious. In this paper, we give an unconditional secure authentication scheme for network coding based on a linear code $C$. Safavi-Naini and Wang \cite{SW} gave an authentication code for multi-receivers and multiple messages. We notice that the scheme of Safavi-Naini and Wang is essentially constructed with Reed-Solomon codes. And we modify their construction slightly to make it serve for authenticating subspace codes over linear network. Also, we generalize the construction with linear codes. The generalization to linear codes has the similar advantages as generalizing Shamir's secret sharing scheme to linear secret sharing sceme based on linear codes \cite{Shamir,MS,M1,M2,CC}. One advantage of this generalization is that for a fixed message space, our scheme allows arbitrarily many receivers to check the integrity of their own messages, while the scheme with Reed-Solomon codes has a constraint on the number of verifying receivers. Another advantage is that we introduce access structure in the generalized scheme. Massey~\cite{M1} characterized the access structure of linear secret sharing scheme by minimal codewords in the dual code whose first component is $1$. We slightly modify the definition of minimal codewords in~\cite{M1}. Let $C$ be a $[V,k]$ linear code. For any coordinate $i\in \{1,2,\cdots,V\}$, a codeword $\vec{c}$ in $C$ is called \emph{minimal respect to $i$} if the codeword $\vec{c}$ has component $1$ at the $i$-th coordinate and there is no other codeword whose $i$-th component is $1$ with support strictly contained in that of $\vec{c}$. Then the security of receiver $R_i$ in our authentication scheme is characterized by the minimal codewords respect to $i$ in the dual code $C^\bot$.

\end{abstract}

\section{Introduction}
\subsection{Background}
Network coding is a novel technique to achieve the maximum multicast throughput, which was introduced by Ahlswede {\it et al.} \cite{ACL}. It allows the intermediate node to generate output data by mixing its received data. In 2003, Li {\it et al.} \cite{LY} further showed that linear network coding is sufficient to achieve the optimal throughput in multicast networks. Subsequently, Ho {\it et al.} \cite{HKM} introduced the concept of random linear network coding, and proved that it achieves the maximum throughput of multicast network with high probability. Network coding is efficiently applicable to numerous forms of network communications, such as Internet TV, wireless networks, content distribution networks and P2P networks. Due to these advantages, network coding attracts many researchers and has developed very quickly.

However, networks using network coding impose security problems that traditional networks do not face. A particularly important problem is the pollution attack. If some nodes in the network are malicious and inject corrupted packets into the information flow, then the honest intermediate node mix invalid packet with other packets.
 According to the rule of network coding, the corrupted outgoing packets quickly pollute the whole network and cause all the messages to be decoded wrongly in the destination.

Recently several related works are proposed to address the pollution attack, such as homomorphic hashing, digital signature and message authentication code (MAC). Krohn {\it et al.} \cite{KFM} (see also \cite{GR}) used homomorphic hashing function to prevent pollution attacks. Yu {\it et al.} \cite{YWR} proposed a homomorphic signature scheme based on discrete logarithm and RSA, which however was showed insecurely by Yun {\it et al.} \cite{YCK}. Charles {\it et al.} \cite{CJL} gave a signature scheme based on Weil pairing over elliptic curves and provided authentication of the data in addition to detecting pollution attacks. Zhao {\it et al.} \cite{ZKM} designed a signature scheme that view all blocks of the file as vectors and make use of the fact that all valid vectors transmitted in the network should belong to the subspace spanned by the original set of vectors from the file. Boneh {\it et al.} \cite{BFK} proposed two signature schemes that can be used in conjunction with network coding to prevent malicious modification of messages, and they showed that their constructions had a lower signature length compared with related prior work. Boneh {\it et al.} \cite{BF} constructed a linearly homomorphic signature scheme that authenticates vectors with coordinates in the binary field $\mathbb{F}_{2}$. It is the first such scheme based on the hard problem of finding short vectors in integer lattices. Agrawal and Boneh \cite{AB} designed a homomorphic MAC system that allows checking the integrity of network coded data. These works provide computational security (i.e., the attacker's resources are limited) in network coding.

Besides digital signatures and MACs, authentication codes also satisfy the properties of authentication. However, authentication code provides unconditional security (i.e., the attacker has unlimited computational power). In the multi-receiver authentication model, a sender broadcasts an authenticated message such that all the receivers can independently verify the authenticity of the message with their own private keys. It requires a security that malicious groups of up to a given size of receivers can not successfully impersonate the transmitter, or substitute a transmitted message. Desmedt {\it et al.} \cite{DFY} gave an authentication scheme of single message for multi-receivers. Safavi-Naini and Wang \cite{SW} extended the DFY scheme \cite{DFY} to be an authentication scheme of multiple messages for multi-receivers. Note that their construction was not linear over the base field with respect to the message. Oggier and Fathi \cite{OF,OF2} made a little modification of the construction so that the construction can be used for network coding. Tang \cite{Tang} used homomorphic authentication codes to sign a subspace which provide an unconditionally security.

In this paper, we consider the combination of authentication code and secret sharing into multicast network coding. And we use subspace codes to transmit messages. The verifying nodes independently verify the authenticity of the message using each own private key, which is distributes by the trusted authority.
Our method of authentication for subspace codes is different from signature through sign a subspace \cite{ZKM,BFK}. Also, compared with the computational security of \cite{ZKM} and \cite{BFK}, our construction is an unconditionally secure authentication code. And compared with the homomorphic scheme \cite{Tang}, our scheme is not homomorphic.

Firstly, we recall the general model of network coding and the definition of subspace codes.
In the basic multicast model for linear network coding, a source node $s$ generates $n$ messages, each consisting of $m$ symbols in the base field $\f$. Let $\{x_{1},x_{2},\ldots,x_{n}\}\subseteq\f^{l\times 1}$ represent the set of messages. Based on the messages, the source node $s$ transmits a message over each outgoing channel. At a node in the network, the symbols on its outgoing channel are $\f$-linear combinations of incoming symbols. For a node $i$, define $Out(i)=\{e\in E : e $ is an outgoing channel of $i\}$, and $In(i) = \{e \in E : e $ is an incoming channel of $i\}$. If the channel $e$ of network carries packet $y(e)$, where $e\in Out(i)$, and $i$ is an internal nodes, then $y(e)$ satisfies $y(e)=\sum_{d\in in(i)}k_{de}y(d)$. The $|In(i)|\times |Out(i)|$ matrix $K_{i}=[k_{de}]_{d\in in(i), e\in Out(i)}$ is called the \emph{local encoding kernel at node} $i$. Note that each $y(e)$ is a linear combination of the messages sent by the source node, so there exists a vector $f_{e}\in \f^{1\times n}$ such that
\begin{equation*}
    y(e)= f_{e}\underline{\mathbf{X}},\,\mbox{where}\,\underline{\mathbf{X}}=\left(
                                                                               \begin{array}{c}
                                                                                 x_1 \\
                                                                                 x_{2}\\
                                                                                 \vdots\\
                                                                                 x_{n} \\
                                                                               \end{array}
                                                                             \right)\ .
\end{equation*}
The vector $f_{e}$ is called the \emph{global encoding vector} of channel $e$. Given the local encoding kernels for all the channels in network, the global encoding kernels can be calculated recursively in any upstream-to-downstream order as follows
$$f_e=\sum_{d\in in(i)}k_{de}f_d\ .$$
 Write the received vectors at a node $t$ as a column vector
 \begin{equation*}
    {A_{t}=(y(e)\ : e\in In(t))^{T}}=\left(
                                       \begin{array}{c}
                                         y(e_{1}) \\
                                         y(e_2) \\
                                         \vdots\\
                                         y(e_{e(t)}) \\
                                       \end{array}
                                     \right)\ ,
 \end{equation*}
 where $In(t)=\{e_{1},e_2,\cdots,e_{e(t)}\}$.
Then we have the decoding equation at the node $t$
 $$ { F_{t}\cdot\mathbf{\underline{X}}=A_{t}\ ,}$$
where
 \begin{equation*}
   F_{t}=(f_{e}: e\in In(t))^{T}=\left(
                                       \begin{array}{c}
                                         f_{e_{1}} \\
                                         f_{e_2} \\
                                         \vdots\\
                                         f_{e_{e(t)}} \\
                                       \end{array}
                                     \right)
 \end{equation*}
is called the \emph{global encoding kernel} at the node $t$.

The set of all subspaces of an $l$-dimensional vector space $\f^l$ forms a projective space $\mathcal{P}_{q}(l)$. The set of all $n$-dimensional subspaces of an $l$-dimensional vector space is called a \emph{Grassmannian manifold} $\mathcal{G}_{q}(l, n)$. A \emph{subspace code} \cite{KK} $\mathcal{C}\subset \mathcal{P}_{q}(l)$ is a collection of subspaces in $\mathcal{P}_{q}(l)$ (for details about subspace codes, one can see \cite{KK}). Moreover, if $\mathcal{C}\subset \mathcal{G}_{q}(l, n)$ then $\mathcal{C}$ is a \emph{constant-dimension code} of dimension $n$ \cite{KK,XF}. For subspace codes, the problem is formulated as transmission of subspaces through a linear network. Suppose the network has minimum cut $n$. Then the transmitter selects a vector space $V\in \mathcal{C}$ from some constant dimension code $\mathcal{C}\subset \mathcal{G}_{q}(l, n)$, and sends a basis of $V$ into the network. The receiver $t$ gathers packets he received, and spans them to form a vector space $U$. Then he regards the subspace $U$ as his received message. It is easy to see that if all channels in the network are error-free, then the node $t$ can decode the original message $V$, i.e., $U=V$ if and only if the global encoding kernel at the node $t$ is of full rank $n$ (also see \cite[Corollary 3]{KK}).

Because the secret key sharing process in our authentication scheme is similar as that in the linear secret sharing scheme, we recall some basic concepts of linear codes and the traditional linear secret sharing scheme.
Let $\f^V$ be the $V$-dimensional vector space over the finite field $\f$ with $q$ elements. For any vector $\vec{x}=(x_1,x_2,\cdots,x_V)\in \f^V$, the \emph{Hamming weight} $\mathrm{Wt}(\vec{x})$ of $\vec{x}$ is defined to be the number of non-zero coordinates, i.e.,
$$\mathrm{Wt}(\vec{x})=\#\left\{i\,|\,1\leqslant i\leqslant V,\,x_i\neq 0\right\}\ .$$
A \emph{linear $[V,k]$ code} $C$ is a $k$-dimensional linear subspace of $\f^V$. The \emph{minimum distance} $d(C)$ of $C$ is the minimum Hamming weight of all non-zero vectors in $C$, i.e.,
$$d(C)=\min\{\mathrm{Wt}(\vec{c})\,|\,\vec{c}\in C\setminus\{\vec{0}\}\}\ .$$
A linear $[V,k]$ code $C\subseteq \f^V$ is called a \emph{$[V,k,d]$ linear code} if $C$ has minimum distance $d$. A vector in $C$ is called a $codeword$ of $C$. A matrix $G\in \f^{k\times V}$ is call a \emph{generator matrix} of $C$ if rows of $G$ form a basis for $C$. A well known trade-off between the parameters of a linear $[V,k,d]$ code is the Singleton bound which states that
$$d\leqslant V-k+1\ .$$
 A $[V,k,d]$ code is called a \emph{maximum distance separable} (MDS) code if $d=V-k+1$. The \emph{dual code} $C^\bot$ of $C$ is defined as the set
\[
  \left\{\vec{x}\in \f^V\,|\,\vec{x}\cdot\vec{c}=0\,\mbox{for all }\vec{c}\in C\right\},
\]
where $\vec{x}\cdot\vec{c}$ is the inner product of vectors $\vec{x}$ and $\vec{c}$, i.e.,
\[
    \vec{x}\cdot\vec{c}=x_1c_1+x_2c_2+\cdots+x_Vc_V\ .
\]

The secret sharing scheme provides security of a secret key by ``splitting" it to several parts which are kept by different persons. In this way, it might need many persons to recover the original key. It can achieve to resist the attack of malicious groups of persons. Shamir \cite{Shamir} used polynomials over finite fields to give an $(S,T)$ threshold secret sharing scheme such that any $T$ persons of the $S$ shares can uniquely determine the secret key but any $T-1$ persons can not get any information of the key. A linear secret sharing scheme based on a linear code~\cite{M1} is constructed as follows: encrypt the secret to be the first coordinate of a codeword and distribute the rest of the codeword (except the first secret coordinate) to the group of shares. McEliece and Sarwate \cite{MS} pointed out that the Shamir's construction is essentially a linear secret sharing scheme based on Reed-Solomon codes. Also as a natural generalization of Shamir'construction, Chen and Cramer \cite{CC} constructed a linear secret sharing scheme based on algebraic geometric codes.

The \emph{qualified subset} of a linear secret sharing scheme is a subset of shares such that the shares in the subset can recover the secret key. A qualified subset is call \emph{minimal} if any share is removed from the qualified subset, the rests cannot recover the secret key. The \emph{access structure} of a linear secret sharing scheme consists of all the minimal qualified subsets. A codeword $\vec{v}$ in a linear code $C$ is said to be
\emph{minimal} if $\vec{v}$ is a non-zero codeword whose leftmost
nonzero component is a $1$ and no other codeword
$\vec{v}^{\prime}$ whose leftmost nonzero component is $1$ has support strictly contained in the support of $\vec{v}$. Massey \cite{M1,M2} showed that the access structure of a linear secret sharing scheme based on a linear code are completely determined by the minimal codewords in the dual code whose first component is $1$.
\begin{prop}[\cite{M1}]\label{minimal}
The access structure of the linear secret-sharing scheme corresponding to
the linear code $C$ is specified by those minimal
codewords in the dual code $C^{\perp}$ whose first component is $1$. In the manner that the set of shares specified by a
minimal codeword whose first component is $1$ in the dual code is the set of shares corresponding
to those locations after the first in the support of this minimal codeword.
\end{prop}

With the above preparation, we next present our construction and main results.
\subsection{Our Construction and Main Results}
Suppose the base field of the network is the finite field $\f$ and we use subspace codes to transmit messages. Take the message space to be the Grassmannian manifold $\mathcal{G}_{q}(l, n)$. The source wants to send a message $U\in \mathcal{G}_{q}(l, n)$, he could send any basis $\vec{s}_1,\vec{s}_2,\cdots,\vec{s}_n$ for $U$. After network coding, any node in the network linearly combines the vectors he received to obtain a linear subspace of $\f^l$. Provided that no error occurs in the network, then the linear subspace is just the message sent by the source if the dimension of the subspace equals $n$. We authenticate the basis $\vec{s}_1,\vec{s}_2,\cdots,\vec{s}_n$. Instead of sending the original base directly, the source node actually sends the authenticated vectors. And each node in the network receives linear combinations of the tagged vectors. Some nodes $R_1,R_2,\cdots,R_V$ in the network can also use their own protocols to verify the integrity of the received vectors. We call these nodes \emph{verifying nodes}.

 There may be some malicious receivers in the network who collude to perform an impersonation attack by constructing a fake message, or a substitution attack by altering the message content such that the new tagged message can be accepted by some other receiver or specific receiver. To substitute the message, the malicious group should generate a vector not in the subspace sent by the source such that the vector with a tag can be accepted by some other receiver.

In this subsection, we present our construction of an authentication scheme based on a linear code for subspace codes in network coding.  It will be shown that the ability of our scheme to resist the attack of the malicious receivers is measured by the minimum distance of the dual code and minimal codewords respect to specific coordinate in the dual code.
\begin{flushleft}
\textbf{Construction:}
\end{flushleft}

Let $C\subseteq \fl^V$ be a linear code with minimum distance $d(C)\geqslant 2$. And assume that the minimum distance of the dual $C^\bot$ is $d(C^\bot)\geqslant 2$. Fix a generator matrix $G$ of $C$

\begin{equation*}
    G=\left(
  \begin{array}{cccc}
    g_{1,1} & g_{1,2} & \cdots & g_{1,V} \\
     g_{2,1} & g_{2,2} & \cdots & g_{2,V} \\
     \vdots & \vdots & \ddots & \vdots \\
    g_{k,1} & g_{k,2} & \cdots & g_{k,V}\\
  \end{array}
\right)\ .
\end{equation*}
Then make $G$ public.

\begin{itemize}
  \item \textbf{Key generation:} A trusted authority randomly chooses parameters
  \begin{equation*}
    A=\left(
  \begin{array}{cccc}
    a_{0,1} & a_{0,2} & \cdots & a_{0,k} \\
     a_{1,1} & a_{1,2} & \cdots & a_{1,k} \\
     \vdots & \vdots & \ddots & \vdots \\
    a_{M,1} & a_{M,2} & \cdots & a_{M,k}\\
  \end{array}
\right)\in \fl^{(M+1)\times k}\ .
\end{equation*}
  \item \textbf{Key distribution:}  The trusted authority computes
    \begin{equation*}
    B=A\cdot G=\left(
  \begin{array}{cccc}
    b_{0,1} & b_{0,2} & \cdots & b_{0,V} \\
     b_{1,1} & b_{1,2} & \cdots & b_{1,V} \\
     \vdots & \vdots & \ddots & \vdots \\
    b_{M,1} & b_{M,2} & \cdots & b_{M,V}\\
  \end{array}
\right)\ .
\end{equation*}
  Then the trusted authority distributes each verifier $R_i$ the $i$-th column of $B$ as the private key, for $i=1,2,\cdots,V$.

  \item \textbf{Authentication tag:} Assume that the source node sends a basis $\vec{s}_1,\vec{s}_2,\cdots,\vec{s}_n\in U$ of a $n$-dimensional subspace $U$ of $\f^l$. The trusted authority chooses an $\f$-linear isomorphism between $\f^l$ and $\fl$. Without any confusion, we identify $\f^l$ with $\fl$ via this isomorphism (this isomorphism is also made public). Then define multiplication in $\f^l$ via the multiplication in $\fl$. The source computes the tag map
      \begin{equation*}
  \begin{array}{cccc}
    L=[L_1,L_2,\cdots,L_k]: & \f^l & \rightarrow & \f^{kl} \\
     & \vec{s} &\mapsto &[L_1(\vec{s}),L_2(\vec{s}),\cdots,L_k(\vec{s})] \ ,
  \end{array}
\end{equation*}
where the map $L_i$ ($i=1,2,\cdots,k$) is defined by\footnote{Note that any $\f$-linear endomorphism $f$ of $\fl$ is of the form
$f(\vec{x})=\sum_{j=1}^l a_i\vec{x}^{q^{j-1}}$ (for any  $\vec{x}\in\fl$)
for some $a_i\in \f, i=1,2,\cdots,l$. To fit the linear operation of the network coding, the tag map should be $\f$-linear or ``near'' $\f$--linear on $\vec{s}$.
}
\[
   L_i(\vec{s})=a_{0,i}+\sum_{j=1}^M a_{j,i}\vec{s}^{q^{j-1}} \mbox{ for any } \vec{s}\in\f^l\ .
\]
For each $i=1,2,\cdots,n$, instead of $\vec{s}_i$, the source node actually sends packets $\vec{x}_i$ of the form
\[
  \vec{x}_i=[1,\vec{s}_i, L(\vec{s}_i)]\in \f^{1+l+kl}\ .
\] \end{itemize}

\begin{rem}
  Add ``$1$'' at the beginning of each tagged vectors, then this scheme can be used to random network coding for keeping the track of the network coding coefficients. In this way, the internal verifying nodes could not know the exact global encoding kernel, but also can do verification of the received vectors. We will see this in the verification step. For network coding with fixed local encoding kernel, we can delete the first bit $1$, and define the tag map $L_i$ to be linear with respect to vector $\vec{s}$, $\sum_{j=0}^M a_{j,i}\vec{s}^{q^{j}}$.
\end{rem}

\begin{flushleft}
\textbf{Verification:}
\end{flushleft}

Suppose the global encoding kernel at the verifying node $R_i$ is
  \begin{equation*}
    H_i=\left(
  \begin{array}{cccc}
    h_{1,1}^{(i)} & h_{1,2}^{(i)} & \cdots & h_{1,n}^{(i)} \\
     h_{2,1}^{(i)} & h_{2,2}^{(i)} & \cdots & h_{2,n}^{(i)} \\
     \vdots & \vdots & \ddots & \vdots \\
    h_{e(i),1}^{(i)} & h_{e(i),2}^{(i)} & \cdots & h_{e(i),n}^{(i)}\\
  \end{array}
\right)\ .
\end{equation*}
Then the node $R_i$ receives the tagged vector
\begin{gather*}
\vec{y}(R_i)=H_i \left(\begin{array}{c}
                   \vec{x}_1 \\
                   \vec{x}_2 \\
                   \vdots \\
                   \vec{x}_n
                 \end{array}\right)\ .
\end{gather*}
The $m$-th row is given by
\[
  \left(\sum_{j=1}^n h_{m,j}^{(i)},\sum_{j=1}^n h_{m,j}^{(i)}\vec{s}_j, \sum_{j=1}^n h_{m,j}^{(i)}L_1(\vec{s}_j),\cdots, \sum_{j=1}^n h_{m,j}^{(i)}L_k(\vec{s}_j)\right)\ .
\]
The verifier $R_i$ checks that whether
\[
 \left(\sum_{j=1}^n h_{m,j}^{(i)}\right)b_{0,i}+\sum_{t=1}^M \left(\sum_{j=1}^n h_{m,j}^{(i)}\vec{s}_j\right)^{q^{t-1}} b_{t,i}
\]
equals to
\[
  \sum_{t=1}^{k} \left(\sum_{j=1}^n h_{m,j}^{(i)}L_t(\vec{s}_j)\right) g_{t,i}\ ,
\]
for all $m=1,2,\cdots,e(i)$.

We call the result $(\sum_{j=1}^n h_{m,j}^{(i)})b_{0,i}+\sum_{t=1}^M \vec{s}^{q^{t-1}} b_{t,i}\in \fl$ the \emph{label} of $R_i$ for $\vec{s}\in\fl$.
\begin{flushleft}
\emph{Correctness of Verification:} if the network has not been invaded, the node $R_i$ should have
\end{flushleft}
\begin{equation*}
    \begin{array}{rl}
       & \sum_{t=1}^{k} \left(\sum_{j=1}^n h_{m,j}^{(i)}L_t(\vec{s}_j)\right) g_{t,i} \\
      = & \sum_{t=1}^{k} \left(\sum_{j=1}^n h_{m,j}^{(i)}(a_{0,t}+\sum_{r=1}^M a_{r,t}\vec{s}_j^{q^{r-1}})\right) g_{t,i} \\
      = & \sum_{j=1}^n h_{m,j}^{(i)}(\sum_{t=1}^{k}a_{0,t}g_{t,i})+\sum_{r=1}^M (\sum_{j=1}^nh_{m,j}^{(i)}\vec{s}_j^{q^{r-1}})(\sum_{t=1}^{k}a_{r,t}g_{t,i})\\
      = & (\sum_{j=1}^n h_{m,j}^{(i)})b_{0,i}+\sum_{r=1}^M \left(\sum_{j=1}^n h_{m,j}^{(i)}\vec{s}_j\right)^{q^{r-1}} b_{r,i}\ .
    \end{array}
\end{equation*}
for all $m=1,2,\cdots,e(i)$.

We summarize the extra costs in general when we communicate messages with the authentication tag:

\begin{tabular}{|c|c|}
  \hline
 \mbox{Tag size} & $kl+1/\f$ \\
\hline
  \mbox{Communication cost} & $kl+1/\f$ \\
\hline
  \mbox{Tag computation cost} & $(M-1)kn \mbox{ exp. }/\fl$ \\
 &  $Mkn \mbox{ multi. }/\fl$\\
\hline
 \mbox{Verification computation cost at $R_i$} & $(M-1)e(i) \mbox{ exp. }/\fl$ \\
 & $(M+k+1)e(i) \mbox{ multi. }/\fl$\\
\hline
  \mbox{Storage at the source} & $(M+1)k/\fl$ \\
\hline
  \mbox{Storage at each verifier} & $M+1/\fl$\\
\hline
  \mbox{Key distribution computation cost} & $(M+1)kV \mbox{ multi. }/\fl$ \\
  \hline
\end{tabular}

Where $1/\f$, $1 \mbox{exp.}/\f$ and $1 \mbox{multi.}/\f$ mean one symbol, one exponent operation and one multiplication operation in the finite field $\f$, respectively.
When we use special generator matrix, e.g., the generator matrix of a systematic code, the cost at the verifying nodes will be less.
Note that the disadvantage is that the tag part introduces much redundancy comparing with the length of original vector.
\begin{flushleft}
\textbf{The Main Results about the Security of Our Scheme:}
\end{flushleft}
The security of the above authentication scheme is summarized in the following theorems.
\begin{thm}\label{thm1}
The scheme we constructed above is an unconditionally secure authentication code for network coding against a coalition of up to $(d(C^\bot)-2)$ malicious receivers.
\end{thm}

The proof of this theorem will be given in Section~2.

More specifically, if we consider what a coalition of malicious receivers can successfully make a substitution attack to one fixed receiver $R_i$. To characterize this malicious group, we slightly modify the definition of minimal codeword in \cite{M1}.
\begin{defn}
Let $C$ be a $[N,k]$ linear code. For any $i\in \{1,2,\cdots,N\}$, a codeword $\vec{c}$ in $C$ is called \emph{minimal respect to $i$} if the codeword $\vec{c}$ has component $1$ at the $i$-th location and there is no other codeword whose $i$-th component is $1$ with support strictly contained in that of $\vec{c}$.
\end{defn}
Similarly as Proposition~\ref{minimal}, we have
\begin{thm}\label{thm2}For the authentication scheme we constructed, we have
\begin{description}
\item[(i)] The set of all minimal malicious groups that can successfully make a substitution attack\footnote{Here, we CLARIFY that in the whole paper, ``can (successfully) make a substitution attack" means that they can make a substitution attack deterministically, and ``can not'' means that they can not successfully produce a fake authenticated message which can be accepted by others in a probability higher than randomly choosing one.} to the receiver $R_{i}$ is determined completely by all the minimal codewords respect to $i$ in the dual code $C^\bot$.

\item[(ii)] All malicious groups that can not produce a fake authenticated message which can be accepted by the receiver $R_{i}$ are one-to-one corresponding to subsets of $[V]\setminus\{i\}$ such that each of them together with $i$ does not contain any support of minimal codeword respect to $i$ in the dual code $C^\bot$, where $[V]=\{1,2,\cdots,V\}$.
\end{description}
\end{thm}
If we take MDS codes, e.g., Reed-Solomon codes, in our construction, Theorems~\ref{thm1} and~\ref{thm2} induces the following corollary.
\begin{cor}\label{cor}
Let $C$ be a $[V,k,d]$ MDS code. For our authentication scheme based on $C$, we have
\begin{description}
\item[(i)] The scheme is an unconditionally secure authentication code for network coding against a coalition of up to ($k-1$) malicious receivers.
\item[(ii)] Moreover, a malicious group can successfully make a substitution attack to any other receiver if and only if the malicious group has at least $k$ members.
\end{description}
\end{cor}

If an authentication scheme satisfies conditions (i) and (ii) in Corollary~\ref{cor}, then we call the authentication scheme a \emph{$(V,k)$ threshold} authentication scheme. In general, it is \textbf{NP}-hard to determine completely that a malicious group can successfully make a substitution attack to others or not. More authentication schemes based on algebraic geometry codes from elliptic curves will be given in Section~3. And we use the group of rational points on the elliptic curve to give a complete classification as Corollary~\ref{cor}.

In Section~2, we explicitly give the security analysis of our authentication scheme, i.e., the proofs of Theorems~\ref{thm1} and ~\ref{thm2}. In Section~3, we give an explicit authentication scheme based on algebraic geometry codes from elliptic curves.
\section{Security Analysis}
In this section, we present the security analysis of our scheme. From the verification step, we notice that the tagged vector $[a,\vec{s},\vec{v}_1,\vec{v}_2,\cdots,\vec{v}_k]$ of one incoming edge can be accepted by the receiver $R_i$, where $a\in\f$ is the corresponding track of the network coding coefficients, if and only if $ab_{0,i}+\sum_{t=1}^{M} \vec{s}^{q^{t-1}} b_{t,i}= \sum_{j=1}^{k} \vec{v}_jg_{j,i}$. So in order to make a substitution attack to $R_i$, it suffices to know the label $ab_{0,i}+\sum_{t=1}^{M} \vec{s}^{q^{t-1}} b_{t,i}$ for some $\vec{s}\in \f^l$ not in the subspace sent by the transmitter, then it is trivial to construct a tag $(\vec{v}_1,\vec{v}_2,\cdots,\vec{v}_k)$ such that $ab_{0,i}+\sum_{t=1}^{M} \vec{s}^{q^{t-1}} b_{t,i}= \sum_{j=1}^{k} \vec{v}_jg_{j,i}$.

The security depends on the hardness to determine the key matrix $A$ or to determine the private key of some other node by solving a system of linear equations. Suppose a group of $K$ malicious nodes collaborate to recover $A$ and make a substitution attack. Without loss of generality, we assume that the malicious nodes are $R_1,R_2,\cdots,R_K$. Each $R_i$ has some information about the key $A$:
\begin{gather*}\label{equation1}
   \left(
     \begin{array}{ccccc}
      \sum_{j=1}^n h_{1,j}^{(i)}& \sum_{j=1}^n h_{1,j}^{(i)}\vec{s}_j & \sum_{j=1}^n h_{1,j}^{(i)}\vec{s}_j^q & \cdots & \sum_{j=1}^n h_{1,j}^{(i)}\vec{s}_j^{q^{M-1}} \\
      \sum_{j=1}^n h_{2,j}^{(i)}&  \sum_{j=1}^n h_{2,j}^{(i)}\vec{s}_j & \sum_{j=1}^n h_{2,j}^{(i)}\vec{s}_j^q & \cdots & \sum_{j=1}^n h_{2,j}^{(i)}\vec{s}_j^{q^{M-1}}\\
      \vdots & \vdots & \vdots & \ddots & \vdots \\
      \sum_{j=1}^n h_{e(i),j}^{(i)}& \sum_{j=1}^n h_{e(i),j}^{(i)}\vec{s}_j & \sum_{j=1}^n h_{e(i),j}^{(i)}\vec{s}_j^q & \cdots & \sum_{j=1}^n h_{e(i),j}^{(i)}\vec{s}_j^{q^{M-1}}\\
     \end{array}
   \right)\cdot A\\=
   \left(
     \begin{array}{cccc}
       \sum_{j=1}^n h_{1,j}^{(i)}L_1(\vec{s}_j) & \sum_{j=1}^n h_{1,j}^{(i)}L_2(\vec{s}_j) & \cdots & \sum_{j=1}^n h_{1,j}^{(i)}L_k(\vec{s}_j) \\
        \sum_{j=1}^n h_{2,j}^{(i)}L_1(\vec{s}_j) & \sum_{j=1}^n h_{2,j}^{(i)}L_2(\vec{s}_j) & \cdots & \sum_{j=1}^n h_{2,j}^{(i)}L_k(\vec{s}_j)\\
       \vdots & \vdots & \ddots & \vdots \\
       \sum_{j=1}^n h_{e(i),j}^{(i)}L_1(\vec{s}_j) & \sum_{j=1}^n h_{e(i),j}^{(i)}L_2(\vec{s}_j) & \cdots & \sum_{j=1}^n h_{e(i),j}^{(i)}L_k(\vec{s}_j)\\
     \end{array}
   \right)
\end{gather*}
and
\begin{gather*}
    A\cdot \left(
                       \begin{array}{c}
                          g_{1,i} \\
                          g_{2,i} \\
                          \vdots\\
                          g_{k,i} \\
                        \end{array}
                      \right)=\left(
                       \begin{array}{c}
                          b_{0,i} \\
                          b_{1,i} \\
                          \vdots\\
                          b_{M,i} \\
                        \end{array}
                      \right)\ .
\end{gather*}
The group of malicious nodes combines their equations, and they get a system of linear equations
\begin{equation}\label{equation}
\left\{
  \begin{array}{c}
    \left(
                     \begin{array}{c}
                       D_1 \\
                       \vdots \\
                       D_K \\
                     \end{array}
                   \right)\cdot A=\left(
                                    \begin{array}{c}
                                      C_1 \\
                                      \vdots \\
                                      C_K \\
                                    \end{array}
                                  \right),
\\
    A\cdot \left(
  \begin{array}{cccc}
    g_{1,1} & g_{1,2} & \cdots & g_{1,K} \\
     g_{2,1} & g_{2,2} & \cdots & g_{2,K} \\
     \vdots & \vdots & \ddots & \vdots \\
    g_{k,1} & g_{k,2} & \cdots & g_{k,K}\\
  \end{array}
\right)=\left(
  \begin{array}{cccc}
    b_{0,1} & b_{0,2} & \cdots & b_{0,K} \\
     b_{1,1} & b_{1,2} & \cdots & b_{1,K} \\
     \vdots & \vdots & \ddots & \vdots \\
    b_{M,1} & b_{M,2} & \cdots & b_{M,K}\\
  \end{array}
\right)\ ,
  \end{array}
\right.
\end{equation}
where
\begin{equation*}
    D_i=\left(
     \begin{array}{ccccc}
      \sum_{j=1}^n h_{1,j}^{(i)}& \sum_{j=1}^n h_{1,j}^{(i)}\vec{s}_j & \sum_{j=1}^n h_{1,j}^{(i)}\vec{s}_j^q & \cdots & \sum_{j=1}^n h_{1,j}^{(i)}\vec{s}_j^{q^{M-1}} \\
       \sum_{j=1}^n h_{2,j}^{(i)}&  \sum_{j=1}^n h_{2,j}^{(i)}\vec{s}_j & \sum_{j=1}^n h_{2,j}^{(i)}\vec{s}_j^q & \cdots & \sum_{j=1}^n h_{2,j}^{(i)}\vec{s}_j^{q^{M-1}}\\
      \vdots & \vdots & \vdots & \ddots & \vdots \\
      \sum_{j=1}^n h_{e(i),j}^{(i)}&  \sum_{j=1}^n h_{e(i),j}^{(i)}\vec{s}_j & \sum_{j=1}^n h_{e(i),j}^{(i)}\vec{s}_j^q & \cdots & \sum_{j=1}^n h_{e(i),j}^{(i)}\vec{s}_j^{q^{M-1}}\\
     \end{array}
   \right)
\end{equation*}
and
\begin{equation*}
C_i= \left(
     \begin{array}{cccc}
       \sum_{j=1}^n h_{1,j}^{(i)}L_1(\vec{s}_j) & \sum_{j=1}^n h_{1,j}^{(i)}L_2(\vec{s}_j) & \cdots & \sum_{j=1}^n h_{1,j}^{(i)}L_k(\vec{s}_j) \\
        \sum_{j=1}^n h_{2,j}^{(i)}L_1(\vec{s}_j) & \sum_{j=1}^n h_{2,j}^{(i)}L_2(\vec{s}_j) & \cdots & \sum_{j=1}^n h_{2,j}^{(i)}L_k(\vec{s}_j)\\
       \vdots & \vdots & \ddots & \vdots \\
       \sum_{j=1}^n h_{e(i),j}^{(i)}L_1(\vec{s}_j) & \sum_{j=1}^n h_{e(i),j}^{(i)}L_2(\vec{s}_j) & \cdots & \sum_{j=1}^n h_{e(i),j}^{(i)}L_k(\vec{s}_j)\\
     \end{array}
   \right)\ .
\end{equation*}
Denote
\begin{equation*}
 S_n=   \left(
      \begin{array}{ccccc}
        1 & \vec{s}_1 & \vec{s}_1^{q} & \cdots & \vec{s}_1^{q^{M-1}} \\
        1 & \vec{s}_2 & \vec{s}_2^{q} & \cdots & \vec{s}_2^{q^{M-1}} \\
        \vdots & \vdots & \vdots & \ddots & \vdots \\
         1 & \vec{s}_n & \vec{s}_n^{q} & \cdots & \vec{s}_n^{q^{M-1}} \\
      \end{array}
    \right) \ .
\end{equation*} Then
\[
  D_i=H_i\cdot S_n\ .
\]

\begin{lem}\label{keylem}
Let $P$ be the subspace of $\fl^k$ generated by $\{g_{i_1},g_{i_2},\cdots,g_{i_K}\}$, where $g_{i_j}$ represents the $i_j$-th column of the
generator matrix $G$. Suppose $K_0=\dim{P}\leqslant k-1$. Then there exists exact $q^{l(M+1-r_0)(k-K_0)}$ matrices $A$ satisfying the system of equations (\ref{equation}), where
  \begin{equation*}
  r_0=\mathrm{rank}\left(
      \begin{array}{c}
        H_{i_1}S_n \\
        H_{i_2}S_n \\
        \vdots \\
        H_{i_K}S_n \\
      \end{array}
    \right)\ .
 \end{equation*}
\end{lem}
\begin{proof} Without loss of generality, we assume $\{i_1,i_2,\cdots,i_K\}=\{1,2,\cdots,K\}$.
Recall the system~(\ref{equation})
\begin{equation*}
\left\{
  \begin{array}{c}
    \left(
                     \begin{array}{c}
                       H_1S_n \\
                       \vdots \\
                       H_KS_n \\
                     \end{array}
                   \right)\cdot A=\left(
                                    \begin{array}{c}
                                      C_1 \\
                                      \vdots \\
                                      C_K \\
                                    \end{array}
                                  \right),
\\
    A\cdot \left(
  \begin{array}{cccc}
    g_{1,1} & g_{1,2} & \cdots & g_{1,K} \\
     g_{2,1} & g_{2,2} & \cdots & g_{2,K} \\
     \vdots & \vdots & \ddots & \vdots \\
    g_{k,1} & g_{k,2} & \cdots & g_{k,K}\\
  \end{array}
\right)=\left(
  \begin{array}{cccc}
    b_{0,1} & b_{0,2} & \cdots & b_{0,K} \\
     b_{1,1} & b_{1,2} & \cdots & b_{1,K} \\
     \vdots & \vdots & \ddots & \vdots \\
    b_{M,1} & b_{M,2} & \cdots & b_{M,K}\\
  \end{array}
\right)\ .
  \end{array}
\right.
\end{equation*}
Rewrite the matrix $A$ of variables as a single column of $k(M+1)$ variables. Then the system~(\ref{equation}) becomes
 \begin{equation}\label{equ2}
    \left(
      \begin{array}{cccc}
        H_1S_n &  &  &  \\
         & H_1S_n &  &  \\
         &  & \ddots &  \\
         &  &  & H_1S_n \\
        \vdots & \vdots & \ddots & \vdots \\
         H_KS_n &  &  &  \\
          & H_KS_n &  &  \\
          &  & \ddots &  \\
          &  &  & H_KS_n \\
        g_{1,1}\vec{I}_{M+1} & g_{2,1}\vec{I}_{M+1} & \cdots & g_{k,1}\vec{I}_{M+1} \\
        g_{1,2}\vec{I}_{M+1} & g_{2,2}\vec{I}_{M+1} & \cdots & g_{k,2}\vec{I}_{M+1} \\
        \vdots & \vdots & \ddots & \vdots \\
        g_{1,K}\vec{I}_{M+1} & g_{2,K}\vec{I}_{M+1} & \cdots & g_{k,K}\vec{I}_{M+1} \\
      \end{array}
    \right)\cdot
    \left(
      \begin{array}{c}
        a_{0,1} \\
        a_{1,1} \\
        \vdots \\
        a_{M,1} \\
        a_{0,2}\\
        a_{1,2}\\
        \vdots \\
        a_{M,2}\\
        \vdots \\
         a_{0,k}\\
         a_{1,k}\\
         \vdots \\
         a_{M,k}\\
      \end{array}
    \right)=T
    \end{equation}
where $\vec{I}_{M+1}$ is the identity matrix with rank ($M+1$) and $T$ is the column vector of the constant terms in system (\ref{equation}) with proper order. Notice that
\begin{equation*}
   r_0= \mathrm{rank}\left(
                         \begin{array}{c}
                           H_1S_n \\
                           H_2S_n \\
                           \vdots \\
                           H_KS_n \\
                         \end{array}
                       \right) =  \mathrm{rank}\left(\left(
                                                       \begin{array}{c}
                                                        H_1 \\
                                                        H_2 \\
                                                        \vdots \\
                                                        H_K \\
                                                       \end{array}
                                                     \right) \cdot S_n \right)
                                                     \leqslant \min\left\{\mathrm{rank}\left(
                                                                                               \begin{array}{c}
                                                                                                H_1 \\
                                                                                                H_2 \\
                                                                                                \vdots \\
                                                                                                H_K \\
                                                                                               \end{array}
                                                                                               \right), n  \right\} \ .
      \end{equation*}
 Also note that rows of
 \begin{equation*}
  \left(
      \begin{array}{c}
        H_1S_n \\
        H_2S_n \\
        \vdots \\
        H_KS_n \\
      \end{array}
    \right)
 \end{equation*}
  is contained in the space $\f^{M+1}$ generated by $g_{i,j}\vec{I}_{M+1}$ if $g_{i,j}\neq 0$.
  So the rank of the big matrix of coefficients equals to
  \[
      r_0k+(M+1-r_0)K_0
  \]
which is less than the number of variables $k(M+1)$. So the system (\ref{equ2}) has
$$q^{l(k(M+1)-( r_0k+(M+1-r_0)K_0))}=q^{l(M+1-r_0)(k-K_0)}$$
solutions, i.e., the system (\ref{equation}) has $q^{l(M+1-r_0)(k-K_0)}$ solutions.
\end{proof}
\begin{rem}
From Lemma \ref{keylem}, in order to cut down the extra costs introduced by authentication, we could choose $M=n$. In this case, $M$ is the minimal integer such that $M\geqslant n$ and Lemma \ref{keylem} holds. Lemma \ref{keylem} is the key lemma in~\cite{zlf} by which we can remove a very important condition in the main result of~\cite{OF2}.
\end{rem}

Note that if $C\,[n,k,d=n-k+1]$ is an MDS code, then whenever $K\leqslant k-1$ the vectors in any $K$-subset of columns of $G$ are linearly independent.

By Lemma \ref{keylem}, the security of our authentication scheme follows.
\begin{thm}\label{general}
The scheme we constructed above with $M=n$ is an unconditionally secure authentication code for network coding against a coalition of up to ($d(C^\bot)-2$) malicious verifiers.
\end{thm}

\begin{proof} Suppose the source node sends $M=n$ $\f$-linearly independent vectors $\vec{s}_1,\vec{s}_2,\cdots,\vec{s}_{M}$, i.e., a basis for a subspace message.
It is enough to consider the case that $K=d(C^\bot)-2$ malicious nodes have received $M$ $\f$-linearly independent vectors $\vec{y}_1,\vec{y}_2,\cdots,\vec{y}_{M}$ and all these malicious nodes are verifying nodes, this is because in this case they know the most information about the key matrix $A$. In other words, the subspace generated by the vectors they received is the subspace sent by the source node. This is also equivalent to the condition: the coalition of global kernels at each malicious node $R_1,\cdots,R_K$ has the rank
\begin{equation*}
  \mathrm{rank}\left(
                       \begin{array}{c}
                        H_1 \\
                        H_2 \\
                        \vdots \\
                        H_K \\
                       \end{array}
                       \right)= n\ .
\end{equation*}
 And under these conditions, they want to make a substitution attack to any other verifying node.

Suppose the malicious nodes $R_1,\cdots,R_K$ want to generate a valid $M$-dimensional subspace $\f\vec{s}'_1\oplus\f\vec{s}'_2\oplus\cdots\oplus\f\vec{s}'_{M}$ such that it can be accepted by $R_{K+1}$. It is equivalent to generating a valid vector $[1,\vec{s}_{M+1},\vec{v}_1,\vec{v}_2,\cdots,\vec{v}_k]$ with
$\vec{s}_{M+1}\notin \f\vec{s}_1\oplus\f\vec{s}_2\oplus\cdots\oplus\f\vec{s}_{M}$
such that it can be accepted by $R_{K+1}$.
 So what they try to do is to guess the label $b_{0,K+1}+b_{1,K+1}\vec{s}_{M+1}+b_{2,K+1}\vec{s}_{M+1}^{q}+\cdots+b_{M,K+1}\vec{s}_{M+1}^{q^{M-1}}$ for some $\vec{s}_{M+1}\notin \f\vec{s}_1\oplus\f\vec{s}_2\oplus\cdots\oplus\f\vec{s}_{M}$ and construct a vector\footnote{This construction step is trivial.} $(\vec{v}_1,\vec{v}_2,\cdots,\vec{v}_k)\in\fl^{k\times 1}$ such that
 \begin{equation*}
\sum_{i=1}^k g_{i,K+1}\vec{v}_i=b_{0,K+1}+b_{1,K+1}\vec{s}_{M+1}+b_{2,K+1}\vec{s}_{M+1}^{q}+\cdots+b_{M,K+1}\vec{s}_{M+1}^{q^{M-1}}\ .
 \end{equation*}
Then the fake message $[1,\vec{s}_{M+1}, \vec{v}_1,\vec{v}_2,\cdots,\vec{v}_k]$ can be accepted by $R_{K+1}$.

In this case, by Lemma \ref{keylem}, there exists $q^{l(k-d(C^\bot)+2)}$ matrices $A$ satisfying the system of equations~(\ref{equation}).

For any $\vec{s}_{M+1}\notin \f\vec{s}_1\oplus\f\vec{s}_2\oplus\cdots\oplus\f\vec{s}_{M}$, we define
\begin{equation*}
    \begin{array}{rccc}
      \varphi_{s_{M+1}}: & \{\textrm{Solutions of System (\ref{equation})}\} & \longrightarrow & \fl \\
       & A & \mapsto & (1,\vec{s}_{M+1},\vec{s}_{M+1}^q,\cdots,\vec{s}_{M+1}^{q^{M-1}})A\left(\begin{array}{c}
                                                                g_{1,K+1} \\
                                                                g_{2,K+1} \\
                                                                \vdots \\
                                                                g_{k,K+1}
                                                              \end{array}\right)\ .
    \end{array}
\end{equation*}
Then we claim:
\begin{description}
\item[(1)] $\varphi_{\vec{s}_{M+1}}$ is surjective.
\item[(2)] for any $y\in \f$, the number of the inverse image of $y$ is $\#\varphi_{\vec{s}_{M+1}}^{-1}(y)=q^{l(k-d(C^\bot)+1)}$.
\end{description}
So the information held by the colluders allows them to calculate $q^l$ equally likely different labels for $s_{M+1}$ and hence their probability of success is $1/q^l$ which is equal to guess a label $b_{0,K+1}+b_{1,K+1}s_{M+1}+b_{2,K+1}s_{M+1}^{q}+\cdots+b_{M,K+1}s_{M+1}^{q^{M-1}}$ for $s_{M+1}$ randomly from $\fl$. And hence we finish the proof of the theorem.

Next, we prove our claim. As $K+1=d(C^\bot)-1$, $g_1,g_2,\cdots,g_{K+1}$ is linearly independent over $\fl$, otherwise the dual code $C^\bot$ will have a codeword with Hamming weight $\leqslant d(C^\bot)-1$ which is impossible by the definition of minimum distance of a code. Then choose $k-K-1=k-d(C^\bot)+1$ extra columns of $G$ such that they combining with $g_1,g_2,\cdots,g_{K+1}$ form a basis of $\fl^k$. Without loss of generality, we assume the first $k$ columns of $G$ is linearly independent of $\fl$. For any $P\in \fl^{(M+1)\times(k-d(C^\bot)+1)}$, consider the system of linear equations
\begin{equation}\label{equ3}
\left\{
  \begin{array}{rl}
  \left(
      \begin{array}{c}
        H_1S_M \\
        H_2S_M \\
        \vdots \\
        H_KS_M \\
      \end{array}
    \right)\cdot A=&
   \left(
                                    \begin{array}{c}
                                      C_1 \\
                                      \vdots \\
                                      C_K \\
                                    \end{array}
                                  \right),
\\
    A\cdot \left(
  \begin{array}{cccc}
    g_{1,1} & g_{1,2} & \cdots & g_{1,K} \\
     g_{2,1} & g_{2,2} & \cdots & g_{2,K} \\
     \vdots & \vdots & \ddots & \vdots \\
    g_{k,1} & g_{k,2} & \cdots & g_{k,K}\\
  \end{array}
\right)=&\left(
  \begin{array}{cccc}
    b_{0,1} & b_{0,2} & \cdots & b_{0,K} \\
     b_{1,1} & b_{1,2} & \cdots & b_{1,K} \\
     \vdots & \vdots & \ddots & \vdots \\
    b_{M,1} & b_{M,2} & \cdots & b_{M,K}\\
  \end{array}
\right),
\\
 A\cdot \left(
  \begin{array}{cccc}
    g_{1,K+2} & g_{1,K+3} & \cdots & g_{1,k} \\
     g_{2,K+2} & g_{2,K+3} & \cdots & g_{2,k} \\
     \vdots & \vdots & \ddots & \vdots \\
    g_{k,K+2} & g_{k,K+3} & \cdots & g_{k,k}\\
  \end{array}
\right)=&P
\ .
  \end{array}
\right.
\end{equation}
By Lemma \ref{keylem}, System (\ref{equ3}) has $q^l$ solutions, saying $A_1,A_2,\cdots,A_{q^l}$. And $A_1,A_2,\cdots,A_{q^l}$ are also solutions of System~ (\ref{equation}). Next, we show
\[
   \left\{\varphi_{\vec{s}_{M+1}}(A_j)\,|\,j=1,2,\cdots,q^l\right\}=\fl\ .
\]
Otherwise, there are two solutions $A_{j_1}$ and $A_{j_2}$ such that
\begin{equation*}
    (1,\vec{s}_{M+1},\vec{s}_{M+1}^q,\cdots,\vec{s}_{M+1}^{q^{M-1}})A_{j_1}\left(\begin{array}{c}
                                                                g_{1,K+1} \\
                                                                g_{2,K+1} \\
                                                                \vdots \\
                                                                g_{k,K+1}
                                                              \end{array}\right)=(1,\vec{s}_{M+1},\vec{s}_{M+1}^q,\cdots,\vec{s}_{M+1}^{q^{M-1}})A_{j_2}\left(\begin{array}{c}
                                                                g_{1,K+1} \\
                                                                g_{2,K+1} \\
                                                                \vdots \\
                                                                g_{k,K+1}
                                                              \end{array}\right)\ .
\end{equation*}
Then we have
\begin{equation*}
\begin{array}{rl}
    & \left(
     \begin{array}{ccccc}
     1&  \vec{s}_1 &\vec{s}_1^q & \cdots & \vec{s}_1^{q^{M-1}} \\
     1&  \vec{s}_2 &\vec{s}_2^q & \cdots & \vec{s}_2^{q^{M-1}} \\
     \vdots&  \vdots & \vdots & \ddots & \vdots \\
     1&  \vec{s}_{M+1} &\vec{s}_{M+1}^q & \cdots & \vec{s}_{M+1}^{q^{M-1}} \\
     \end{array}
   \right) A_{j_1} \left(
  \begin{array}{cccc}
    g_{1,1} & g_{1,2} & \cdots & g_{1,k} \\
     g_{2,1} & g_{2,2} & \cdots & g_{2,k} \\
     \vdots & \vdots & \ddots & \vdots \\
    g_{k,1} & g_{k,2} & \cdots & g_{k,k}\\
  \end{array}
\right)\\
= &
     \left(
     \begin{array}{ccccc}
     1&  \vec{s}_1 &\vec{s}_1^q & \cdots & \vec{s}_1^{q^{M-1}} \\
     1&  \vec{s}_2 &\vec{s}_2^q & \cdots & \vec{s}_2^{q^{M-1}} \\
     \vdots&  \vdots & \vdots & \ddots & \vdots \\
     1&  \vec{s}_{M+1} &\vec{s}_{M+1}^q & \cdots & \vec{s}_{M+1}^{q^{M-1}} \\
     \end{array}
   \right) A_{j_2} \left(
  \begin{array}{cccc}
    g_{1,1} & g_{1,2} & \cdots & g_{1,k} \\
     g_{2,1} & g_{2,2} & \cdots & g_{2,k} \\
     \vdots & \vdots & \ddots & \vdots \\
    g_{k,1} & g_{k,2} & \cdots & g_{k,k}\\
  \end{array}
\right) \ .
\end{array}
\end{equation*}
But the Moore matrix
\begin{equation*}
     \left(
     \begin{array}{ccccc}
     1&  \vec{s}_1 &\vec{s}_1^q & \cdots & \vec{s}_1^{q^{M-1}} \\
     1&  \vec{s}_2 &\vec{s}_2^q & \cdots & \vec{s}_2^{q^{M-1}} \\
     \vdots&  \vdots & \vdots & \ddots & \vdots \\
     1&  \vec{s}_{M+1} &\vec{s}_{M+1}^q & \cdots & \vec{s}_{M+1}^{q^{M-1}} \\
     \end{array}
   \right)
   \end{equation*}
 is invertible since $\vec{s}_1,\vec{s}_2,\cdots,\vec{s}_{M+1}\in\fl$ are linearly independent over $\f$. And the matrix
  \begin{equation*}
 \left(
  \begin{array}{cccc}
    g_{1,1} & g_{1,2} & \cdots & g_{1,k} \\
     g_{2,1} & g_{2,2} & \cdots & g_{2,k} \\
     \vdots & \vdots & \ddots & \vdots \\
    g_{k,1} & g_{k,2} & \cdots & g_{k,k}\\
  \end{array}
\right)
\end{equation*}
is invertible by our assumption. So $A_{j_1}=A_{j_2}$ which contradicts to the condition $A_{j_1}\neq A_{j_2}$. And hence, the statement (1) holds.

Next, we prove (2).
Any one solution of System~(\ref{equation}) gives one $P\in \fl^{(M+1)\times(k-d(C^\bot)+1)}$, while corresponding to such a $P$ there are $q^l$ solutions of System~(\ref{equation}) from the proof of (1). In this way, we partition solutions of System~(\ref{equation}) into $q^{l(k-d(C^\bot)+1)}$ parts such that each part contains $q^l$ elements. Also from the proof of (1), the image of each part under $\varphi_{\vec{s}_{M+1}}$ is $\fl$. So for any $y\in \fl$, the number of the inverse image of $y$ is $\#\varphi_{\vec{s}_{M+1}}^{-1}(y)=q^{l(k-d(C^\bot)+1)}$. So far, we have finished the proof of our claim.

\end{proof}

\begin{rem}\label{motivation}
From the proofs of Lemma \ref{keylem} and Theorem \ref{general}, the coalition of malicious nodes $B$ can successfully make a substitution attack to the node $R_i$ if and only if $g_i$ is contained in the subspace of $\fl^k$ generated by $\{g_j\,|\,j\in B\}$, where $g_j$ represents the $j$-th column of the
generator matrix $G$. In this case, they can recover the private key of $R_i$ using the linearity relationship. So we connect our authentication scheme with the linear secret sharing scheme in the way that we regard the private key of $R_i$ as the secret key and the private keys of other verifying nodes are shares of the private key of $R_i$. Then similarly as the linear secret sharing scheme \cite{M1,M2} which considered the first component of codewords as the secret key location, using the modified definition of the minimal codewords of a linear code given in the introducion, we can characterize the malicious groups that can successfully make a substitution attack to some other node completely. This is what Theorem~\ref{thm2} and Corollary~\ref{cor} say about.
\end{rem}

\section{The Authentication Scheme Based on Algebraic Geometry Codes}
In this section, we give examples of our authentication schemes based on some explicit linear codes, AG codes from elliptic curves. First, recall the definition of AG codes.

We fix some notation valid for this entire section.
\begin{itemize}
\item\emph{
 $X/\f$ is a geometrically irreducible smooth projective curve of genus $g$  over the finite field $\f$
with function field $\f(X)$. }
\item \emph{ $X(\f)$ is the set of all $\f$-rational points on $X$.}
\item\emph{ $D=\{R_{1},R_{2},\cdots,R_{n}\}$ is a proper subset of
rational points $X(\f)$.}
\item\emph{Without any confusion, also write
$D=R_{1}+R_{2}+\cdots+R_{n}$.}
\item\emph{ $G$ is a divisor of degree $k$ ($2g-2<k<n$) with
$\mathrm{Supp}(G)\cap D=\emptyset$.}
\end{itemize}
Let $V$  be a divisor on $X$. Denote by $\mathscr{L}(V)$ the
$\f$-vector space of all rational functions $f\in \f(X)$ with the
principal divisor $\mathrm{div}(f)\geqslant -V$, together with the
zero function. And Denote by $\Omega(V)$ the $\f$-vector space of
all Weil differentials $\omega$ with divisor
$\mathrm{div}(\omega)\geqslant V$, together with the zero
differential (cf.~\cite{Sti}).

Then the residue AG code $C_{\Omega}(D, G)$ is defined to be the
image of the following residue map:
\[
    res: \Omega(G-D)\rightarrow \f^{n};\, \omega\mapsto
     (res_{R_{1}}(\omega),res_{R_{2}}(\omega),\cdots,res_{R_{n}}(\omega))\enspace .
\]
 And its dual
code, the functional AG code $C_{\mathscr{L}}(D, G)$ is defined to
be the image of the following evaluation map:
\[
       ev: \mathscr{L}(G)\rightarrow \f^{n};\, f\mapsto
       (f(R_{1}),f(R_{2}),\cdots,f(R_{n}))\enspace .
\]

They have the code parameters $[n, n-k+g-1, d\geqslant k-2g+2]$ and
$[n, k-g+1, d\geqslant n-k]$, respectively. And we have the following isomorphism
\[
    C_{\Omega}(D, G)\cong C_{\mathscr{L}}(D, D-G+(\eta))
\]
for some Weil differential $\eta$ satisfying
$\upsilon_{P_{i}}(\eta)=-1$ and $\eta_{P_{i}}(1)=1$ for all
$i=0,1,2,\cdots,n$ (\cite[Proposition~2.2.10]{Sti}).

For the authentication scheme based on the simplest AG codes, i.e., generalized Reed-Solomon codes, we have determined all the malicious groups that can make a substitution attack to any (not necessarily all) other in Corollary~\ref{cor}. Next, we consider the authentication scheme based on AG codes $C_{\Omega}(D, G)$ from elliptic curves. Using the Riemann-Roch theorem, the malicious groups who together are able to make a substitution attack to any (not necessarily all) other or not can be characterized completely as follows.
\begin{thm}\label{esss}
Let $X=E$ be an elliptic curve over $\f$,
$D=\{R_{1},R_{2},\cdots,R_{n}\}$ a  subset of $E(\f)$ such that the
zero element $O\notin D$ and let $G=kO$ ($0<k<n$). Then for the authentication scheme we constructed based on the AG code $C_{\Omega}(D, G)$, we have
\begin{description}
\item[(i)] Any coalition of up to $(n-k-2)$ malicious receivers can not make a substitution attack to any other receiver.
\item[(ii)] A malicious group $A\subseteq D$, $\#A=n-k-1$, can successfully make a substitution attack to the receiver $R_j\in D\setminus A$ if and only
if
\[
      \sum_{P\in D\setminus A}P=R_j\enspace .
\]
Moreover, we note that they can only successfully make a substitution attack to the receiver $\sum_{P\in D\setminus A}P$ if $\sum_{P\in D\setminus A}P\in D\setminus A$.
 \item[(iii)] A malicious group $A\subseteq D$, $\#A=n-k$, can successfully make a substitution attack to the receiver $R_j\in D\setminus A$ if and
only if
there exists some $Q\in E(\f)\setminus\{R_j\}$ such that the sum
\[
      Q+\sum_{P\in D\setminus A}P=R_j\enspace ,
\]
which is equivalent to
\[
 \sum_{P\in D\setminus A}P\neq O\ .
\]
And hence, such a malicious group can successfully make a substitution attack to any other receiver.
\item[(iv)] A malicious group with at least $(n-k+1)$ members can successfully make a substitution attack to any other receiver.
\end{description}
\end{thm}
\begin{proof}
The statement (i) follows from Theorem~\ref{general} as the minimum distance
\[
  d^{\bot}(C_{\Omega}(D, G))=d(C_{\mathscr{L}}(D, G))\geq n-k\ .
\]

For the statement (ii), if the malicious group $A\subseteq D$, $\#A=n-k-1$, can successfully make a substitution attack to the receiver $R_j\in D\setminus A$, then there exists some non-zero function in the dual code $f\in \mathscr{L}(kO-\sum_{R\in D\setminus A}R+R_j)$, i.e.,
\[
  \mathrm{div}(f)\geqslant \sum_{R\in D\setminus A}R-R_j-kO\ .
\]
Both sides of the above inequality have degree $0$, so
\[
   \mathrm{div}(f)= \sum_{R\in D\setminus A}R-R_j-kO\ .
\]
That is,
\[
    \sum_{R\in D\setminus A}R=R_j\ .
\]

Similarly for the statement (iii), if a malicious group $A\subseteq D$, $\#A=n-k$, can successfully make a substitution attack to the receiver $R_j\in D\setminus A$, there exists some non-zero function $f\in \mathscr{L}(kO-\sum_{R\in D\setminus A}R+R_j)\setminus \mathscr{L}(kO-\sum_{R\in D\setminus A}R) $, i.e.,
\[
  f(R_j)\neq 0\,\mbox{ and }\, \mathrm{div}(f)\geqslant \sum_{R\in D\setminus A}R-R_j-kO\ .
\]
Then there is an extra zero $Q\in E(\f)\setminus\{R_j\}$ of $f$ such that
\[
  \mathrm{div}(f)= \sum_{R\in D\setminus A}R-R_j+Q-kO\ .
\]
That is,
\[
    \sum_{R\in A}R+Q=R_j\ .
\]
The rest of (iii) is obvious.

We prove the statement (iv) by contradiction. A malicious group $A$ can not successfully make a substitution attack to the receiver $R_j$ if and only if
there exists a linear function
\[
    f\in \mathscr{L}(D-G+(\eta))
\]
such that
\[
   f(R_j)=1, \mbox{ and }\ f(R)=0\ \forall R\in A\ .
\]
As $f\in \mathscr{L}(D-G+(\eta))$, $f$ has at most $\deg(D-G+(\eta))=n-k$ zeros. So if
\[
   \# A\geq n-k+1\ ,
\]
the malicious group $A$ can successfully make a substitution attack to any other receiver.


\end{proof}

Finally, we give a remark on the above theorem to finish this section.
\begin{rem}
If for any $A\subseteq D$ with $\#A=n-k$, the inequality
\[
 \sum_{P\in D\setminus A}P\neq O\ .
\]
holds, then the minimum distance~\cite{cheng,zfw} of the AG code $C_{\Omega}(D, G)$ is
\[
    d(C_{\Omega}(D, G))=k+1\ .
\]
In this case, $C_{\Omega}(D, G)$ is MDS. So by the property of MDS codes, its dual code $C_{\mathscr{L}}(D, G)$ is also MDS, i.e.,
\[
  d(C_{\mathscr{L}}(D, G))=n-k+1\ .
\]
Also in this case, such a malicious group in Theorem \ref{esss}(ii) does not exist. So it coincides with Corollary~\ref{cor}.

On the other side, if $C_{\Omega}(D, G)$ is not MDS, then there exists $A\subseteq D$ with $\#A=n-k$ such that
\[
 \sum_{P\in D\setminus A}P= O\ .
\]
Such a malicious group $A$ can not successfully make a substitution attack to any other receiver.
\end{rem}

\section{Conclusion}
In this paper, we construct an authentication scheme based on linear code $C\,[V,k,d]$ for subspace codes over network coding. It is an unconditional secure authentication scheme, which can offer robustness against a coalition of up to ($d(C^\bot)-2$) malicious receivers. If we take $C$ to be Reed-Solomon codes, then our authentication scheme can be regarded as a modification of the multi-receiver authentication scheme for multiple messages given by Safavi-Naini and Wang~\cite{SW}. The authentication scheme based on the Reed-Solomon code $[V,k,d]$ is a $(V,k)$ threshold authentication scheme, any $k-1$ of the $V$ receivers can not produce a fake message, with a higher probability than randomly guessing a label for the message, that can be accepted by any other receiver, but any $k$ of the $V$ verifying receivers can easily produce a fake message that can be accepted by any other receiver. To generalize the scheme with Reed-Solomon codes to that with arbitrary linear codes, there are several advantages similar as the advantages of generalizing Shamir's secret sharing scheme to linear secret sharing sceme \cite{Shamir,MS,M1,M2,CC}. First, for a fixed message space $\mathcal{G}_{q}(l, n)$, by choosing proper linear codes, our scheme allows arbitrarily many receivers to check the integrity of their own messages. while the scheme with Reed-Solomon codes has a constraint on the number of verifying receivers $V\leqslant q^l$. Secondly, for some important receiver, coalitions of $k$ or more malicious receivers can not yet make a substitution attack on the receiver more efficiently than randomly guessing a label from the finite field for a fake message.


\bibliographystyle{ieeetran}
\bibliography{authentication}

\end{document}